\documentclass[10pt, compsocconf, conference]{IEEEtran}
\IEEEoverridecommandlockouts

\usepackage{epsfig,subfigure,latexsym,amssymb,dsfont,cite,algorithmic,kbordermatrix,multirow,float,lipsum,bm,stfloats}        

\usepackage[namelimits]{amsmath} 
\usepackage{amsfonts}            
\usepackage{mathrsfs}            

\usepackage[linesnumbered,ruled,vlined]{algorithm2e}
\usepackage[table]{xcolor}

\newtheorem{theorem}{Theorem}
\newtheorem{lemma}{Lemma}

\newenvironment{proof}%
{\begin{trivlist}\item[\hspace*{\labelsep}{\it Proof.\/}]}%
{\hfill$\Box$\end{trivlist}}

\newtheorem{defi}{Definition}

\newcommand{\head}[1]
 {\markright{\hbox to 0pt{\vtop to 0pt{\hbox{}\vskip 3mm \hrule
 width  \textwidth \vss} \hss}{\sc #1}}}
 
\begin{document}
\title{A Truthful FPTAS Mechanism for Emergency Demand Response in Colocation Data Centers}

\author{\authorblockN{Jianhai Chen\authorrefmark{1}, Deshi Ye\authorrefmark{1}, Shouling Ji\authorrefmark{1}\authorrefmark{2}, Qinming He\authorrefmark{1}, Yang Xiang\authorrefmark{3}, Zhenguang Liu\authorrefmark{4}}
\authorblockA{\authorrefmark{1}  Institute of Cyberspace Research and College of Computer Science and Technology, Zhejiang University, Hangzhou 310027.\\
\{chenjh919, yedeshi, hqm, sji\}@zju.edu.cn
}
\authorblockA{\authorrefmark{2} Alibaba-Zhejiang University Joint Institute of Frontier Technologies, Hangzhou 310027, China.}
\authorblockA{\authorrefmark{3}Swinburne Research, Swinburne University of Technology, Victoria 3122 Australia.\\
yxiang@swin.edu.au 
}
\authorblockA{\authorrefmark{4}Dept. of Computer Science, 
National University of Singapore, Computing 1, 13 Computing Drive Singapore 117417.\\
liuzhenguang2008@gmail.com
}
\thanks{Remark: this paper is published as main conference paper in IEEE INFOCOM 2019.}
}

\maketitle
\begin{abstract}
Demand response (DR) is not only a crucial solution to the demand side management but also a vital means of electricity market in maintaining power grid reliability, sustainability and stability. DR can enable consumers (e.g. data centers) to reduce their electricity consumption when the supply of electricity is a shortage. The consumers will be rewarded in the case of DR if they reduce or shift some of their energy usage during peak hours. Aiming at solving the efficiency of DR, in this paper, we present MEDR, a mechanism on emergency DR in colocation data center. First, we formalize the MEDR problem and propose a dynamic programming to solve the optimization version of the problem. We then design a deterministic mechanism as a solution to solve the MEDR problem. We show that our proposed mechanism is truthful. Next, we prove that our mechanism is an FPTAS, {\it i.e.}, it can be approximated within $1 + \epsilon$ for any given $\epsilon > 0$, while the running time of our mechanism is polynomial in $n$ and $1/\epsilon$, where $n$ is the number of tenants in the datacenter. Furthermore, we also give an auction system covering the efficient FPTAS algorithm as bidding decision program for DR in colocation datacenter. Finally, we choose a practical smart grid dataset to build a large number of datasets for simulation in performance evaluation. By evaluating metrics of the approximation ratio of our mechanism, the non-negative utility of tenants and social cost of colocation datacenter, the results demonstrate the effectiveness of our work.

\end{abstract}

\begin{IEEEkeywords}
Emergency Demand Response; Knapsack Problem; Mechanism Design; Colocation Data Center; Auction
\end{IEEEkeywords}
\section{Introduction}
Demand response~(DR) programs are widely adopted in many colocation data centers for improving the efficiency of power grids~\cite{Kamyab2016Demand,Moghaddam2017On}. 
It attempts to adjust the demand for power instead of adjusting the supply. In today's power grid, DR is a technique for regulating the energy consumption over time, which is one of the major reliability impacts for smart grids~\cite{Wei2017Energy}. 
On the consumer side, DR is commonly utilized as a powerful tool for employing flexibility of using electricity in response to supply-demand conditions\cite{zhang2015infocom}. In smart grid market when electricity price rises or the system reliability are threatened, the electricity supplier will firstly deliver the notice of direct compensation, of inductively reducing power load or signal of power price rise to electricity consumers. The consumers will change their intrinsic power using mode so as to meet the demand of electricity supply, reducing or passing a special period of time of power load, ensuring the stability of the power grid and restraining the rise of electricity price. In addition, DR is an important tool of demand-side management~\cite{Jacquot2018Demand} which refers to the way that countries use policy measures to guide power users to reduce electricity at peak time, use electricity in low valley, improve power supply efficiency and optimize the usage of electricity. When it comes to the emergency demand of using electricity, the demand-side management will be required to start an immediate response or incentive mechanism of the user report electricity, and declare the amount of electricity and the price of electricity. 

However, there are few efficient mechanisms and systems to support efficient power management in current power grid and many colocation data centers, leading to high power cost and low efficiency~\cite{Jacquot2018Demand}. For example, to reduce peak demand in a power grid, the DR is usually implemented manually by sending signals to large consumers, such as data centers. 
Besides, it is worthy to note that colocation of data centers is quite popular now. According to the website\footnote{Data collected from http://www.datacentermap.com/ in Jan 7, 2016.}, there are $3,775$ colocation data centers from $112$ countries. A colocation is the third-party leased placement that provides physical homes for many data centers, and provides lots of services such as fast Internet, stable power supply and cooling. 
Though colocation of data centers provided a nice solution for those enterprise tenants, it consumed huge electricity. As pointed out in~\cite{zhou2015demand}, $91$ billion kilowatt-hours of electricity was consumed in U.S. in 2013, and it emitted around $97$ million metric tons of carbon pollution in that year. On the other hand, it is possible to close or migrate some tasks in a large data center such that some computing servers can be shutdown. This makes possible for data centers to be participant in DR. In case of emergency (for example, earthquake or extreme bad weather) or reaching the capacity of a grid, it requires to implement the DR.

Therefore, to improve the efficiency, an incentive reverse auction mechanism is employed to motivate power users to participate in DR activities. The basic process of the auction is as follows. The electricity DR and electricity quota are issued by the special power management department. After the power users receive the signal, the electricity price can be reported in the form of bidding. Then the bidding system runs a bidding mechanism to choose the group of users, in which users will perform the restricted electricity consumption behavior in the DR and achieve corresponding compensations.

In this work, we aim to design an efficient mechanism to solve the DR problem called MEDR. Actually, there are several challenges in solving MEDR. It is worth to note that we could apply the VCG mechanism~\cite{vickrey1961counterspeculation,clarke1971multipart,groves1973incentives} to MEDR problem, and a deterministic truthful mechanism was obtained. However, the optimization problem of MEDR is NP-hard, since it generalizes the min-knapsack problem. On the other hand, approximation algorithms for MEDR might not be truthful. To the best of our knowledge, Zhang et al.~\cite{zhang2015infocom} was the first one to study approximated truthful mechanisms for MEDR problem. 
They provided a 2-approximated mechanism with truthful in expectation. The main technique of their work is to present a 2-approximated algorithm, and then turn the approximation algorithm into a mechanism with truthful in expectation, while keeping the approximation ratio of 2. The framework of their work is based on a convex decomposition technique~\cite{lavi2011truthful}, which will transfer an approximation algorithm into a truthful randomized mechanism.

Our main contribution is to propose a deterministic truthful mechanism with FPTAS approximated. To the best of our knowledge, we are not aware of other deterministic approximated truthful mechanisms. The main technique of our work is to design a monotone algorithm based on the framework of Archer and Tard\"os~\cite{archer2001truthful}, which is the key idea of designing deterministic truthful mechanism for one parameter. In detail, we first design a dynamic programming for the optimization version of our studied problem, and then applying some rounding technique such that our design dynamic programming satisfies the monotone property. Furthermore, in our mechanism, both the cost $c_i$ and the size $e_i$ of each agent $i$ are private information. We also implement all the algorithms in a reverse auction system and a simulation tool for performance evaluation. The system tool can be used for colocation data centers or can be extended to be adopted to some other applications such as auctions in smart grid. Extensive experiments are presented to evaluate the effectiveness of our method. Empirical results show that our mechanism achieves nearly optimal solution.

The rest of this paper is organized as follows. In Section~\ref{sec:ps}, we state the mechanism design problem called MEDR. In Section~\ref{sec:sm}, we propose a dynamic programming to solve the optimization version of the studied problem. Then we address an FPTAS deterministic truthful mechanism. In Section~\ref{sec:exp}, we implement a reverse auction system tool and all the algorithms. Extensive simulation experiments are taken to evaluate the effectiveness of our method. The related works are presented in Section~\ref{sec:rw}, and concluding remarks are given in Section~\ref{sec:conc}.

\section{Problem statement}\label{sec:ps}

In this section we address the statement of a mechanism design problem on DR for colocation data center.

We study a mechanism design problem MEDR that arises in data center demand response. There are $n$ tenants in a colocation data center.  Each tenant $i \in \{1, 2, \ldots, n\}$ subscribes a certain amount of power supply from the colocation operator. 
In the event of Emergency Demand Response (EDR), the colocation operator is required to reduce $W$ amount of energy. 
Given power-based contracts, tenants may not have incentive to participate in EDR unless they are awarded. Even if some tenants are interested in EDR, their reduction may not reach the reduction target $W$. In case of not reaching the target, the colocation operator can use backup energy storage (BES) to fulfill the shortage of EDR target. Let $y$ be the amount of grid-power demand reduction due to the usage of BES, and $\alpha$ be the cost of BES usage per kWh.

Each tenant $i$ submits a bid with two parameters $(s_i, b_i)$, 
where $s_i$ is the amount of planned energy reduction and $b_i$ is the claimed cost due to such a reduction. 
However, each tenant $i$ has its own true type $(e_i, c_i)$, where $c_i$ is the cost due to a reduction of $e_i$ energy. The value $e_i$ and $c_i$ are only known to the tenant $i$. Moreover, each tenant is {\em single-minded}~\cite{lehmann2002truth} such that each tenant is restricted to one single bid.
Every tenant has freedom to choose participation in this EDR or not. 
If a tenant is not willing to participate in this EDR, we can suppose its bid is $(0, 0)$. 
Let $B =\{(s_1, b_1), \ldots, (s_n, b_n)\}$ be the set of bids by the  $n$ tenants.
Based on this bidding $B$, the colocation operator will pay money $P_i(B)$ to each tenant $i$ to encourage their participate in this EDR.
Let $U_i(B) = P_i(B) - c_i$ be the utility of tenant $i$ according to the biddings of $B$. 
Clearly each tenant $i$ attempts to maximize his/her utility. 
According to~\cite{zhang2015infocom}, the power consumption in colocation data center consists of both the energy consumption of tenants and also consumption of management such as cooling. 
There is a ratio called Power Usage Effectiveness (PUE) $\gamma$ between the total energy consumption to the energy consumed by tenants,  which typically ranges from 1.1 to 2.0.

A tenant is a {\em winner} if her/his bidding is successful. Let $N$ be the set of winners. To meet the energy reduction target $W$, we require that $y + \gamma \sum_{i\in N} s_i  \geq W$.
The social cost of the colocation operator is $\alpha y + \sum_{i \in N} P_i(B)$. The social cost of tenants is $\sum_{i \in N} (c_i - P_i(B))$. Thus, the total social cost is equivalent to aggregate tenant cost due to energy reduction plus the operator's cost for using BES, i.e., $\alpha y + \sum_{i\in N} c_i$. 
The goal of the mechanism design is to minimize the total social cost, meanwhile no tenant can benefit by proposing a false bidding. The optimization version of this problem can be formulated as an integer programming. Let $x_i = 1$ if tenant $i$ is a winner, i.e. $i \in N$, otherwise $x_i =0$. 

\begin{equation}\label{eq:prob}
    \begin{array}{rcl}
       min \ \  \ \ \alpha y + \sum_{i=1}^{n} x_i  c_i    
    \end{array}
\end{equation}
subject to
\begin{equation}\label{equ:eigen11}
    \begin{array}{rcl}
    y + \gamma \sum_{i=1}^{n} x_i s_i  \geq W       
    \end{array}
\end{equation}
\begin{equation}\label{equ:eigen12}
    \begin{array}{rcl}
    x_{i}= \{1, 0\},   \forall  i \in \{1, 2, \ldots, n\}      
    \end{array}
\end{equation}

The studied MEDR problem is closely related to knapsack auction problems. 
According to objective functions, we define two types of mechanism design for knapsack problems. In {\em max-knapsack} problem, each agent has a private valuation for having his/her objective in the knapsack. The problem is to find an allocation of the agents without exceeding the capacity of the knapsack as so to maximize the sum of each agent's value.In {\em min-knapsack} problem, each agent has a private cost for having
his/her item in the knapsack. The problem is to find an allocation to cover the knapsack, while the sum of agents' cost is minimized. 

For any instance $I$, we define by $\mathcal{C}(\mathcal{M}(I))$ the {\em social cost} of the mechanism $\mathcal{M}$, which is the total costs of tenants plus the operator's cost for using BES. 
A mechanism $\mathcal{M}$ is said to be $\rho$-approximated if $\mathcal{C}(\mathcal{M}(I)) \leq \rho \cdot \mathcal{C}(OPT(I))$, where $OPT$ is an optimal algorithm. 

Let $B_{-i} =\{B_1, \ldots, B_{i-1}, B_{i+1}, \ldots, B_n\}$ be the bids except tenant $i$'s bid.

\begin{defi}\label{def:truth}
(Truthfulness): A mechanism $\mathcal{M}$ consisting of an allocation function $\mathcal{A}$ and a payment function $\mathcal{P}$ is truthful (or strategy-proof) if for every tenant $i$ with the true cost $c_i$ cannot increase his/her utility by declaring any other cost $(s_i, b_i)$ regardless of every bidding of other agents $B_{-i}$, i.e., it satisfies 
\[
U_i((e_i,c_i), B_{-i}) \ge U_i((s_i,b_i), B_{-i}).
\]
\end{defi}
This definition implies that truthful reporting is a dominant strategy for every tenant.

\begin{defi}\label{def:ir}
(Individual rationality): A mechanism $\mathcal{M}$ is said to be individual rationality if every agent always obtains non-negative utility with bidding of the true cost, i.e., $U_i((e_i,c_i), B_{-i}) \geq 0$ for any $i$ and any $B_{-i}$.
\end{defi}

\section{Approximated Truthful Mechanism}\label{sec:sm}

In this section, we will address an approximated truthful mechanism. We present a dynamic programming to solve the optimization version of our problem MEDR optimally, and then we explore a deterministic truthful mechanism  while it is arbitrarily approximated for any given $\epsilon >0$.

\subsection{Dynamic Programming Model}\label{sec:dp}

Our dynamic programming requires to solve the min-knapsack problem as a subroutine. 
The min-knapsack problem consists in finding a subset of items, where each item $i$ has a cost $c_i$ and a size $s_i$, with the minimum cost such that the sum of their sizes is at least as large as a specified capacity. Based on the idea of the max-knapsack problem~\cite{lawler1979fast}, Tauhidul~\cite{tauhidul2009approximation} gave a dynamic programming for the min-knapsack problem. We adopt this dynamic programming in (\ref{opt:mk})~\cite{tauhidul2009approximation} as a subroutine in the following.

Let $S(i, c)$ denote a subset of $\{1, \ldots, i\}$ whose cost is exactly $c$ and whose total size is maximized. 
Let $A(i, c)$ be the size of $S(i, c)$ ($A(i, c) = 0$ if no such set exists). The recursive formula of the dynamic programming is given in (\ref{opt:mk}). In this formula $A(i, c)$ gives a tabular of an optimal value for each subproblem $(i, c)$.

\begin{eqnarray}\label{opt:mk}
A(i, c) = 
\left\{
 \begin{array}{ll}
   \max \{A(i-1, c), s_{i} + A(i-1, c - c_{i})\},  \\
     \ \ \ \ if \; c_{i} \leq c \\
    A(i-1, c), \ \ \ \   otherwise
 \end{array}
\right.
\end{eqnarray}

In the following, we design the dynamic programming for our MEDR problem based on the recursive function (\ref{opt:mk}).
The Algorithm~\ref{alg:dyedr} (Algorithm $DOPT(I)$) gives the details of the dynamic programming for our problem MEDR. 

\begin{algorithm}
\caption{Algorithm $DOPT(I)$: Dynamic Programming for MEDR}\label{alg:dyedr}
\KwIn{The set of tenants $I$, and demand capacity $W$.}
Run the dynamic programming based on the formula (\ref{opt:mk}) for the input $I$, and obtain $A(i, c)$ for each $(i, c)$, where $1\le i\leq n$ and $0\le c \le \sum_i c_i$;

\For{each $(i, c)$}{
  \eIf{$\gamma A(i, c) < W$}{
     $y(i, c) = \alpha (W - \gamma A(i, c)) + c$
  }
  {
  $y(i, c) = c$
  }
  }
\KwOut{Return $\min_{(i, c)} y(i, c)$}
\end{algorithm}


\begin{theorem}\label{theo:dyopt}
The dynamic programming $DOPT(I)$ produces an optimal solution for any instance of tenants $I$ with demand request $W$, and unit cost of BES $\alpha$. The running time is pseudo-polynomial, which is $O(n^2 c_{\max})$, where $c_{\max}$ is the largest cost.
\end{theorem}
\begin{proof}
Any optimal solution consists of two parts, one is covered by BES, and another is covered by items from $I$. Let $p, q$ be the cost due to tenants $I$ and BES, respectively. Let $c_{\max} =\max_i c_i$ be the largest cost among all tenants.

In the dynamic programming we iterate all possible $(i, c)$, where $c \in C =\{0, 1, 2, \ldots, nc_{\max}\}$. Any cost due to tenants $I$ is in $C$, hence $p \in C$. Note that 
the dynamic programming $A(i, p)$ provides the maximize size whose cost is exactly $p$. If $\gamma A(n, p)$ is less than $W$, and we require at least $W - \gamma A(n, p)$ BES to cover the knapsack in the optimal. Therefore, $p + q \geq \alpha(W - \gamma A(n, p)) + p$. If $\gamma A(n, p) \geq W$, then $q=0$.
These two cases are both covered in the dynamic programming, which implies that the dynamic programming outputs an optimal solution.

The running time of dynamic programming is $O(n^2 c_{\max})$, since $i \leq n$, and $c \leq nc_{\max}$, and the running time is bounded by the iterative function of $(i, c)$.
\end{proof}

\subsection{Monotone FPTAS Model}\label{sec:fptas}

%


Motivated by the truthful mechanism for max-knapsack problem~\cite{briest2011approximation}, we will propose a deterministic truthful mechanism. To keep the truthful property, the idea of our mechanism is to give a monotone algorithm. To obtain an FPTAS, we need to design a monotone algorithm whose approximation ratio is arbitrarily close to 1. The detailed algorithm is given in Algorithm~\ref{alg:FPTAS}, which  iteratively calls a subroutine Algorithm~\ref{alg:AK} (Algorithm~$A_r(k,I)$). The motivation of Algorithm~\ref{alg:AK} is to keep the truthful, in which the rounding of each item is independent on the bidding of all tenants. 

\begin{algorithm}
\caption{Algorithm~$A_r(k,I)$}\label{alg:AK}
\KwIn{Given parameter $k$, and the instance $I = (s_1, c_1), \ldots, (s_n, c_n)$.}

Let $a_k = \frac{\epsilon 2^k}{n +1}$;

Let $T(k)$ be the subset of items with cost at most of $2^k$; We construct a new instance $I^\prime$ based on $T(k)$, which is identical to $T(k)$, but the cost of each item $c^\prime$ is given as below.

\For{$i \in T(k)$} {
  $c_i^\prime = \lfloor \frac{c_i}{a_k} \rfloor$
  }


Run the dynamic programming $DOPT(I^\prime)$for the items in $T(k)$ with cost $c_i^\prime$, and obtain $A(i, c^\prime)$;

\For{each $(i, c^\prime)$}{
  \eIf{$\gamma A(i, c^\prime) < W$}{
     $y(i, c^\prime) = \lfloor \frac{\alpha (W - \gamma A(i, c^\prime))}{a_k} \rfloor + c^\prime$
  }
  {
  $y(i, c^\prime) = c^\prime$
  }
  }
\KwOut{Return $\min_{(i, c^\prime)} y(i, c^\prime)$}

\end{algorithm}


\begin{algorithm}
\caption{Monotone FPTAS $A_{FPTAS}$}\label{alg:FPTAS}
\KwIn{Given $\epsilon >0$, and the instance $I$.}

Let $best \gets \infty$, and $c_{\max} = \max_i c_i$.

\For{$k \gets 1$ \textbf{to} $\log c_{\max}$} {
  $S^\prime(k) \gets A_{r}(k, I)$; /* call Algorithm~\ref{alg:AK} (Algorithm $A_r(k, I)$) */

  \If{$S^\prime(k) < best$}{
  $best \gets S^\prime(k)$

  $\bar{S} \gets$ the subset items that contained in the solution of $S^\prime(k)$
  }
  }
\KwOut{$\bar{S}$, and use BES $W - \gamma \sum_{i \in \bar{S}} s_i$}
\end{algorithm}

\begin{lemma}\label{lem:fptas}
For any $\epsilon >0$, Algorithm $A_{FPTAS}$ has approximation ratio of $1 + \epsilon$, and its running time is polynomial in $1/\epsilon, n, \log c_{max} $.
\end{lemma}
\begin{proof}
Let $c_q$ be the largest cost among the items in an optimal algorithm to cover the knapsack.
Define $k^*$, such that 
\[
 2^{k^* -1} < c_q \leq 2^{k^*}.
\]

Denote $O^*$ to be the subset of items in the optimal solution. Let $y^*$ be the size BES used in the optimal solution. Let $O^*(R) = O^* \bigcup \{R\}$, where $R$ is a virtual item with size $y^*$ and cost $\alpha y^*$. Let $OPT$ be the cost of the optimal solution. We have $OPT \geq c_q$.

Note that in $T(k^*)$ as denoted in the algorithm $A_r(k,I)$, we have $O^* \subseteq T(k^*)$. 
Let $\bar{S}$ be the subset of items returned by the algorithm $A_r(k,I)$ with $k^*$ as the parameter, and let $(\bar{i}, \bar{c})$ be the pair of values that reaches the minimum of $A_r(k,I)$. 

Let $O^\prime$ be the subset of items with costs rounded by $2^{k^*}$ from $O^*$. Let $R^\prime$ be a virtual item with size $y^*/a_{k^*}$.

Let $ALG$ be the final cost incurred by the algorithm $A_{FPTAS}$, 
we can use the following inequalities to approximate the cost by the algorithm with the optimal solution. 

\begin{eqnarray*}
ALG & = & \sum_{i \in \bar{S}} c_i + \max(\alpha (W - \gamma A(\bar{i}, \bar{c})), 0) \\
 & \leq & \sum_{i \in \bar{S}} c_i + \max(\lfloor \frac{\alpha (W - \gamma A(\bar{i}, \bar{c}))}{a_k^*} \rfloor, 0) a_{k^*}  + a_{k^*}\\
    & \leq & \sum_{i \in \bar{S}}(c_{i}^\prime \cdot a_{k^*} + a_{k^*}) + \\
    & & \max(\lfloor \frac{\alpha (W - \gamma A(\bar{i}, \bar{c}))}{a_k^*} \rfloor, 0) \cdot a_{k^*}  + a_{k^*}\\
    & \leq & \sum_{i \in \bar{S}}c_{i}^\prime \cdot a_{k^*} + \max(\lfloor \frac{\alpha (W - \gamma A(\bar{i}, \bar{c}))}{a_k^*} \rfloor, 0) \cdot a_{k^*} \\
    & &  + (n+ 1) a_{k^*}\\ 
    & \leq & \sum_{i \in O^\prime \bigcup \{R^\prime\}}c_{i}^\prime \cdot a_{k^*} + (n+1) a_{k^*} \\
    & \leq & \sum_{i \in O^* \bigcup \{R\}}c_{i} + (n+1) a_{k^*} \\
    & \leq & OPT + \epsilon 2^{k^*} \\
    & \leq & (1+2\epsilon) OPT.
\end{eqnarray*}

The running time is $poly(1/\epsilon, n, \log c_{\max})$.
In algorithm $A_r(k,I)$, the largest cost of $T(k)$ is $2^k$, the running time of dynamic programming is $O(n^3/\epsilon)$. The total running time of $A_{FPTAS}$ is $O(\frac{1}{\epsilon}n^3 \log c_{max})$.
\end{proof}

\subsubsection{Monotone}

A declaration $B_i^\prime =(s_i^\prime, b_i^\prime)$ is said to be a {\em higher declaration} than the bidding $B_i = (s_i, b_i)$ if $s_i^\prime \geq s_i$ and $b_i^\prime \leq b_i$, i.e. $B_i \preceq B_i^\prime$. A bid $(s_i, b_i)$ is said to be a {\em winner declaration} if this item is selected in the knapsack.

\begin{defi} (Monotone)
We say that an algorithm $A$ is monotone if, for any bidder $(s_i, b_i)$ is a winning declaration then any higher declaration also wins. 
\end{defi}

Bitonic was introduced by Mu'Alem and Nisan~\cite{mu2008truthful} for maximize problems, such as multi-unit auction, and it was generalized by Briest, Krysta, and V{\"o}cking~\cite{briest2011approximation}. 

In this work, we apply the technique of bitonic to the minimize problems. 

\begin{defi} (Bitonic)
Given a function $f: \mathcal{A}^n \rightarrow $, a monotone algorithm $A$ is bitonic with respect to the function $f$ if for any agent $i$, the following hold:
\begin{enumerate} 
\item If $i \in A(B)$, then $f(A(B_i, B_{-i})) \geq f(A(B_i^\prime, B_{-i}))$ for any 
$B_i  \preceq B_i^\prime$.
\item If $i \not\in A(B)$, then $f(A(B_i, B_{-i})) \geq f(A(B_i^\prime, B_{-i}))$ for any 
$B_i^\prime  \preceq B_i$.
\end{enumerate}
\end{defi}

Intuitively, a monotone algorithm $A$ is bitonic with respect to $f$ if $f$ is a monotone non-decreasing function of each agent's valuation while she is not selected for the solution, but becomes monotone non-increasing after she is selected for the solution. In this work, the function $f$ is the objective function, i.e., the social welfare. The bitonic is indeed required to guarantee the monotone for compositions of algorithms.  

\begin{algorithm}
\caption{$MIN(A_1, A_2)$ Operator}\label{alg:min2}
\KwIn{Bidding $B$}
Run the algorithm $A_1$ and $A_2$; 

Let $sw_{A_1}(B)$ and $sw_{A_2}(B)$ be the social welfare of Algorithm $A_1$ and $A_2$, respectively.

\eIf {$sw_{A_1}(B) \leq sw_{A_2}(B)$} { 
  return $A_1(B)$;}
  {
   return $A_2(B)$.
  }
\end{algorithm}

\begin{lemma}\label{lem:minop}
Let $A_1$ and $A_2$ be two monotone bitonic allocation algorithms. 
Then, $M = MIN(A_1,A_2)$ is a monotone bitonic allocation algorithm.
\end{lemma}
\begin{proof} This can be easily extended from the proof of the Theorem 2 in~\cite{mu2008truthful}, which was designed for the MAX operator. 
\end{proof}

\begin{lemma}\label{lem:akm}
Algorithm $A_r(k,I)$ is monotone and bitonic with respect to the objective function.
\end{lemma}
\begin{proof}
Algorithm $A_r(k,I)$ returns an optimal solution, if an agent reports a higher bidder, then the optimal algorithm will accept this item too. Suppose an agent $i$ was not selected, and any lower declaration $B_i^\prime$, if this item was accepted then the objective function shall be smaller, otherwise the objective remains, and hence the objective function is non-increasing for any lower bidders. Thus the property of bitonic follows.
\end{proof}

\begin{lemma}\label{lem:asb}
Algorithm $A_{FPTAS}$ is monotone and bitonic with respect to the objective function.
\end{lemma}
\begin{proof}
The lemma follows immediately according to Lemma~\ref{lem:minop} and Lemma~\ref{lem:akm}.
\end{proof}

\subsubsection{Payment}

\begin{defi}(Critical payment)\label{def:cp}
Let algorithm $A$ be a monotone algorithm, if we fix the declaration $B_{-i}$, and then for any agent $i$ and fixed bidding $s_i$, there exists a unique cost $\theta_i^A$, called {\it critical payment}, such that $\forall b_i \leq \theta_i^A$, $b_i$ is a winning declaration, and $\forall b_i > \theta_i^A$ is a losing declaration.
\end{defi}

To calculate the critical value for any agent $j$, we fix the other agents' bids, and then use a binary search on interval $[b_j, \max_j b_{j}]$ and repeatedly run the allocation algorithm $A$ to check whether the agent $j$ is selected.
 

\begin{defi}\label{defi:pay}
The payment $p^A$ associated with the monotone allocation algorithm $A$ that is based on the critical value is defined by $p_j^A = \theta _{j}^{A}$ if agent $j$ wins with allocation $Alloc_i(B) = s_i$, and $p_j^A = 0$ otherwise.
\end{defi}

A mechanism $M_A = (A, p^A)$ is {\em normalized }, if its payment $p^A$ is defined as in Definition~\ref{defi:pay}, i.e. agents that are not selected pay $0$. 
We say that algorithm $A$ is {\em exact} if $Alloc_i(B) = s_i$ or $Alloc_i(B)= \emptyset$ for each declaration $(s_i, b_i)$. 

In this work, we only consider a limited type of agent called {\em single-minded}, the cost function 
$\infty$ if $Alloc_i(B) > s_i$ and $c_i$ otherwise. That will force each agent does not over bidding his/her size if   allocation algorithm is an exact algorithm.

\begin{theorem}~\cite{briest2011approximation}\label{theo:msp}
Let $A$ be a monotone and exact algorithm for some minimization problems and single-minded agents. Then mechanism $M_A = (A, p^A )$ is truthful.
\end{theorem}
\begin{proof}
In~\cite{briest2011approximation}, they gave the detailed proof for utilitarian problems, and thus it holds for MAX-knapsack problem. Moreover, in their paper, it was shown that the proof is valid for minimization problems, such as the reverse single-minded multi-unit auction problem, which is equivalent to the minimum knapsack problem. 
\end{proof}

\begin{algorithm}
\caption{Algorithm $P^{\mathcal{A}}(B)$}\label{alg:payment}
\KwIn{The bidding $B$ of all tenants, and the allocation algorithm ${\mathcal{A}}$}


\For{$i \gets 1$ \textbf{to} $n$} {
  
  Let $z_i  = 1$ if the $i$th item is selected in the knapsack by the allocation problem ${\mathcal{A}} (B)$, and $0$ otherwise.

  Let $h = \alpha \gamma s_i$ and $l = b_i$.

  \While {$h - l \ge 1$}{
  $b_i^\prime = (h + l)/2$;

  $z_i = {\mathcal{A}} (B_{-i}, (s_i, b_i^\prime))$;

   \eIf {$z_i == 1$}{ 
   $l = b_i^\prime$
   }  
   { 
   $h = b_i^\prime$
   }    
  }
 $P_i(B) \gets l$.
  }
\KwOut{The payment $P_i(B)$ for each agent $i$.}
\end{algorithm}

\begin{theorem}\label{theo:trfptas}
The mechanism $M_{A_{FPTAS}} = (A_{FPTAS}, p^{A_{FPTAS}})$ is truthful and it is an FPTAS mechanism, i.e. its approximation ratio is $1+\epsilon$ for any given $\epsilon > 0$, and the total running time of the mechanism is polynomial in $1/\epsilon, n, \log c_{max} $.
\end{theorem}
\begin{proof}
Algorithm $P^{\mathcal{A}}(B)$ (Algorithm~\ref{alg:payment}) is a critical payment, Algorithm $A_{FPTAS}$ is an exact algorithm, and bitonic with respect to the objective function according to Lemma~\ref{lem:asb}. Thus the mechanism $M_{A_{FPTAS}}$ is truthful followed by Theorem~\ref{theo:msp}. The FPTAS is achieved in Lemma~\ref{lem:fptas}.

\end{proof}

\section{Implementation}

In this section, we provide a reverse auction system~(RAS) to solve the MEDR problem and implement all the algorithms for DR on colocation data centers. 

\subsection{Reverse Auction System}

The architecture of the RAS is shown in Fig.~\ref{fig:ras}. 
The whole system framework consists of two parts, namely, client and server. The client refers to users including the bidding tenants and the colocation data center operator. 
The server is a MEDR auction system, which is implemented and run as a service. This service program involves three primary function modules, that is, {\em auction controller~(AC)}, {\em bidding decision~(BD)} and {\em power controller (PC)}, respectively. 

The {\em AC} module is responsible for interaction with client users, including receiving the client tenants' bidding application, invoking bidding decision (BD) to make a bidding decision, and returning back the bidding decision results to client users.
The {\em BD} module implements all the MEDR algorithms presented in Section~\ref{sec:sm}, including the dynamic programming algorithm $A_{DOPT}$, the monotone FPTAS algorithm $A_{FPTAS}$ and the payment algorithm $A_{Payment}$. All the algorithms are developed by C/C++ language, and are integrated together as an independent service program. 
The {\em PC} module takes charge of executing power supply policies for balancing power using according to the bidding results. In particular, the RAS runs in a colocation data center for all tenants and the colocation operator. 

\begin{figure}[ht]
\centering
\includegraphics[width=9.0cm,height=8.0cm]{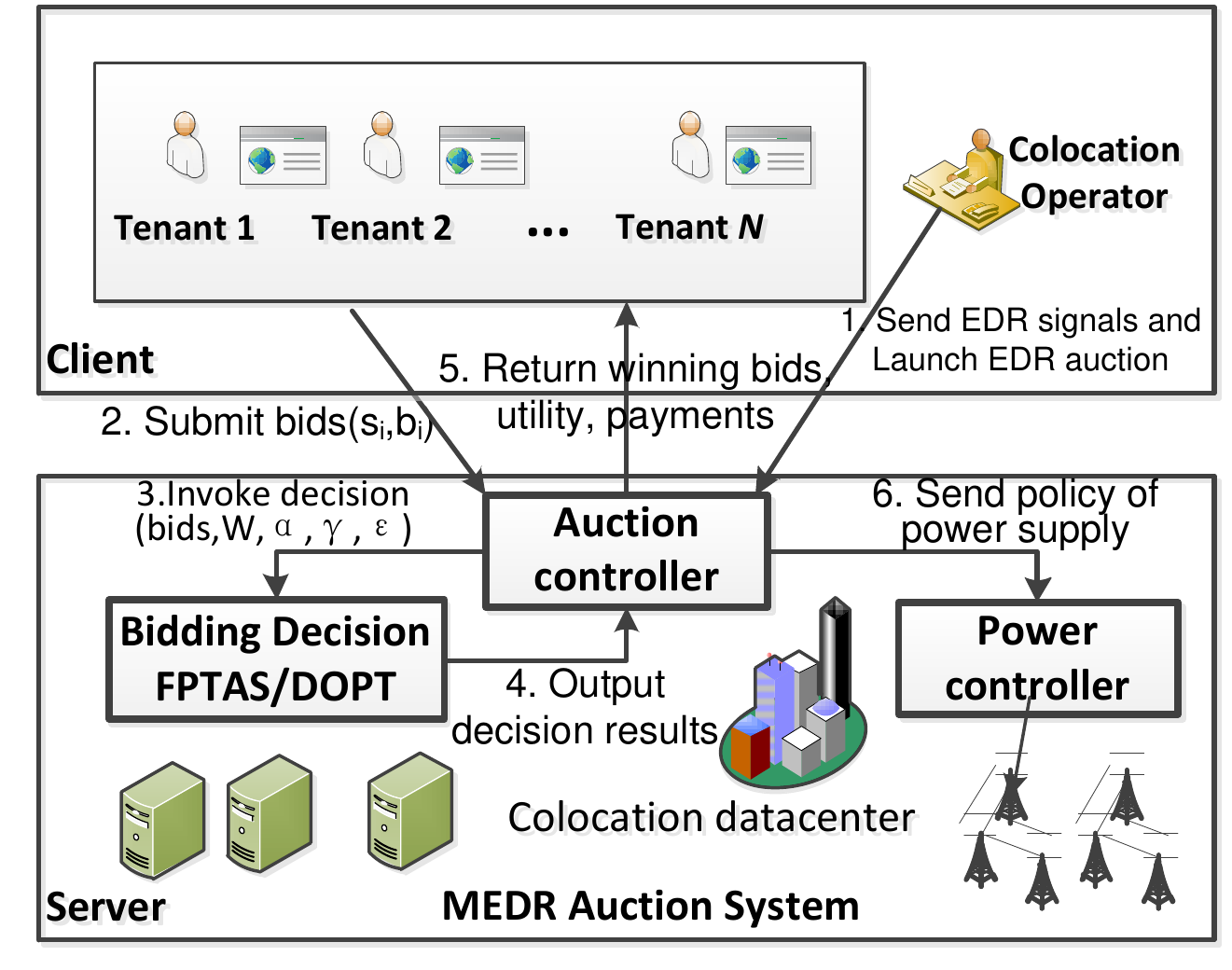}
\caption{Architecture of MEDR auction system }
\label{fig:ras}
\end{figure}

\subsection{Reverse Auction Process}

As shown in Fig.~\ref{fig:ras}, the basis process of reverse auction is as follows. At first, upon arrival of DR or EDR signals in client, the colocation operator (Auctioneer) will firstly launch an auction in RAS and send the EDR signal to tenants. The tenants who are willing to participate then submit their own bids with the size of power and costs to the RAS server. 

Then, the \textbf{AC} receives the bids from client tenant users, and then ask for BD to process the bidding decision when the system needs to make decision of bidding. In bidding decision, BD invokes FPTAS to make a decision. Before invoking, BD needs to input the bidding data and the value of decision algorithm required parameters, such as EDR $W$, $\alpha$, $\gamma$, and $\epsilon$, etc. 
BD will input all data from the AC and calls a decision algorithm to get bidding decision results. After running the decision algorithms, it outputs the bidding decision results to the AC.

Further, the AC returns back the bidding decision results to the client. The results include the winning bids, and payments. 
In the end, not only the winning tenants will achieve compensation rewards from the colocation data center operator, but also the new policy of power supply will be submitted to the PC and be executed for the DR.

\section{Performance evaluation}\label{sec:exp}

In this section, we firstly generate simulation datasets. We then do holistic simulation experiments to evaluate the efficiency of FPTAS mechanism which further validates the effectiveness of our RAS system performance.

\subsection{Simulation Datasets}

We consider a colocation data center, with $N(N=9)$ participating tenants (denoted as Tenant \#1, Tenant \#2, ...,Tenant \#9). 
Each tenant $i$ has $m_{i}(=100,000)$ homogeneous servers. The power of each server includes $150W$ computing power and $d_0=100W$ idle power. That is to say, if a server is turned off, then we can obtain 100W power. If a server is running a workload, it needs $150+100=250W$ power. 

For simplicity, we produce simulation datasets according to the datasets used in~\cite{zhang2015infocom},
including total EDR energy reduction by PJM on January 7, 2014 and normalized workload traces which are collected from~\cite{Lin2013Dynamic}~(``Hotmail'' and ``MSR'') and~\cite{Vasudevan2009Safe}~(``Wikipedia''), as shown in Fig.~\ref{fig:edra} and Fig.~\ref{fig:wload}.
The EDR energy reduction data consists of eleven event numbers. Each number represents an energy reduction of demand response arise in one hour period of time. The eleven time periods are made up of a range in hour from 5 to 11 and from 16 to 19. 
Moreover, we get the specific data from the two figures and list in Table~\ref{tab:wl} in detail. Like Zhang~\cite{zhang2015infocom}, we further generate EDR energy reduction targets as $15\%$  of the total EDR, which lists in the $W$ column of Table~\ref{tab:wl}.

The tenants' bids data are generated randomly with consideration of the workload traces and some data center factors, such as the general cost of a server or whole data center and the energy price in market. We duplicate each workload for all tenants with randomness up to $20\%$. For the total $9$ tenants, we firstly generate $9$ random numbers $r_{i}$, where $i=1,2,\ldots,9$, between $0.01$ and $0.20$. In our simulation, the random numbers are $0.11, 0.06, 0.02, 0.13, 0.05, 0.01, 0.15, 0.16, 0.02$. 
Then we duplicate the workload for three tenants by Hotmail, three by MSR and three by Wikipedia, respectively. We assume that each size of the workload denotes a ratio of the total number of servers which is in running status. If we turn off some servers then we can slash energy consumption. The total energy reduction by tenant $i$ is $s_{i}=n_{i}.d_{0,i}.T$, where $T=1$ hour in simulation. 

We conclude the size of tenants' bid by a formula: 
$s_{i}=workloads * M * r_{i} * d_{0}/1,000,000 (MW)$, where $M=100,000$ and $d_{0}=100$. The final range of size $s$ is between $10\sim 80$. The workloads use data shown in Table~\ref{tab:wl}. Besides, according to Zhang~\cite{zhang2015infocom}, tenants can reasonably save $6.7\sim 13.3$ cents/kWh (depending on electricity price) power when they house servers in their data centers. Equivalently, tenants can save $67\sim 133$ \$/MWh. Hence we produce several random numbers $rb$ between 67\$ and 133\$ for each tenant as a bid power price which tenants are willing to take part in the auction activity. The $i$th tenant bid is concluded by formula $b_{i}=s_{i} * rb_{i}$. All produced simulation data are in Table~\ref{tab:sm}. Each tenant bid $B{{i}}$ has two numbers $s_{i},b_{i}$, which denote a size of energy reduction and cost for supplying its responding size of power, respectively. 

\begin{figure}[t]
  \centering
  \mbox{
      \subfigure[Total EDR energy reduction by PJMon January 7, 2014.]{
        \includegraphics[width=0.24\textwidth,height=0.20\textwidth]{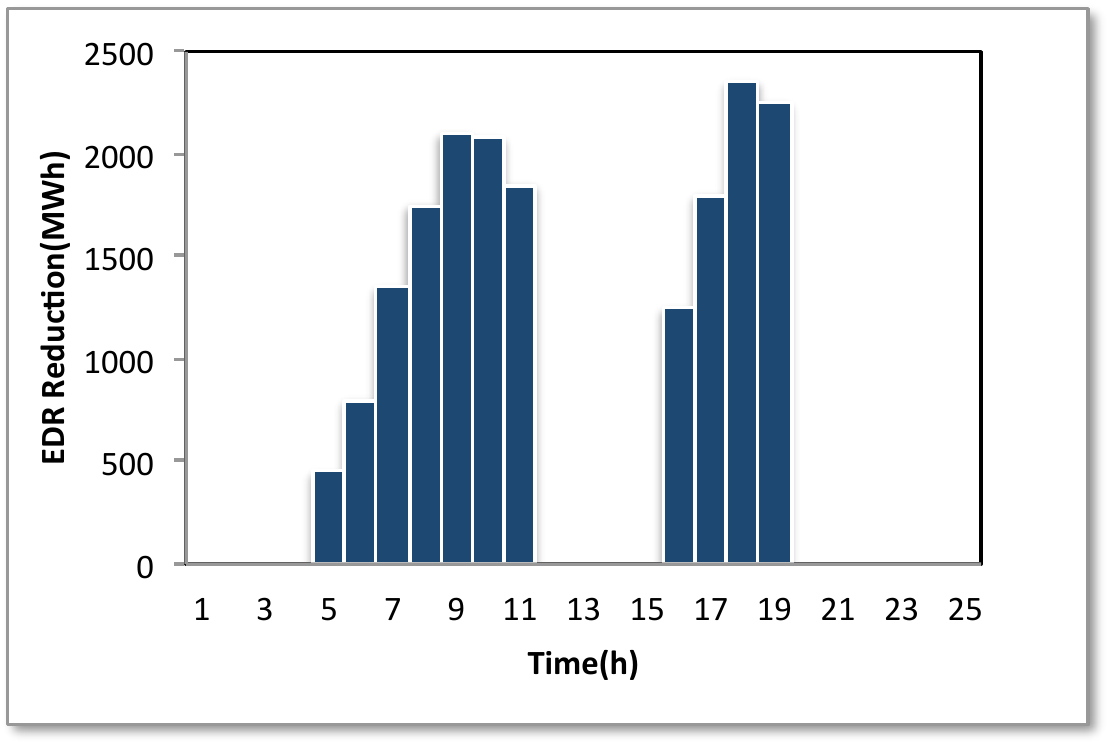}\label{fig:edra}}\hspace{.02in}
      \subfigure[Normalized workloads in distinct time.] {
        \includegraphics[width=0.23\textwidth,height=0.20\textwidth]{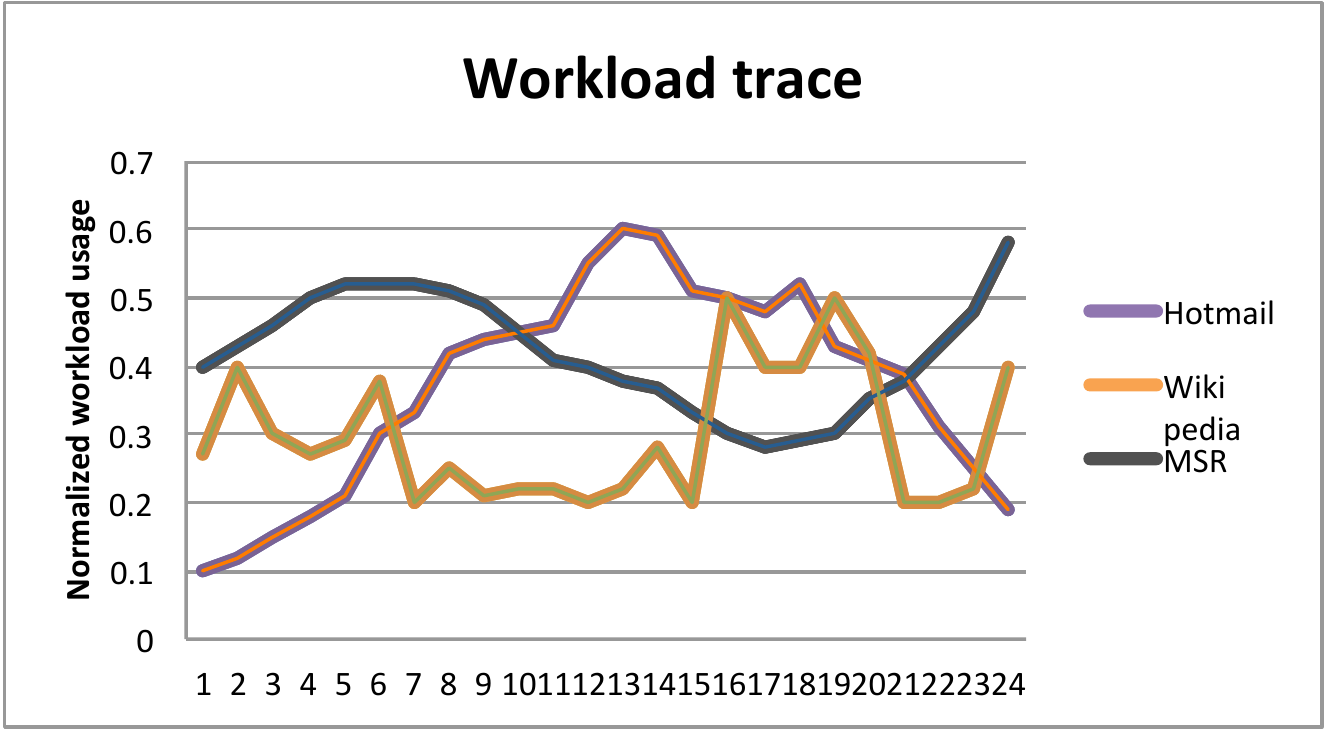}\label{fig:wload}}
       
  }
 
  \caption{The source dataset for bidding decision simulation.}
\end{figure}

\begin{table}[ht]
\caption{EDR target \& Workload}\label{tab:wl}
\centering
\begin{tabular}{|c|c|c|c|c|c|}
\hline
\multirow{1}{*}{Time(h)}&\multirow{1}{*}{EDR}&\multirow{1}{*}{W(15\%EDR)}
&\multirow{1}{*}{Hotmail}& \multirow{1}{*}{MSR}&\multirow{1}{*}{Wikipedia}
\\
\hline
\multirow{1}{*}{5}&\multirow{1}{*}{450}&\multirow{1}{*}{68}
&\multirow{1}{*}{0.21}&\multirow{1}{*}{0.52}&\multirow{1}{*}{0.29}
\\ 
\hline
\multirow{1}{*}{6}&\multirow{1}{*}{800}&\multirow{1}{*}{120}
&\multirow{1}{*}{0.30}&\multirow{1}{*}{0.52}&\multirow{1}{*}{0.20}
\\ 
\hline
\multirow{1}{*}{7}&\multirow{1}{*}{1,350}&\multirow{1}{*}{203}
&\multirow{1}{*}{0.33}&\multirow{1}{*}{0.52}&\multirow{1}{*}{0.20}
\\
\hline
\multirow{1}{*}{8}&\multirow{1}{*}{1,750}&\multirow{1}{*}{263}
&\multirow{1}{*}{0.42}&\multirow{1}{*}{0.51}&\multirow{1}{*}{0.25}
\\
\hline
\multirow{1}{*}{9}&\multirow{1}{*}{2,100}&\multirow{1}{*}{315}
&\multirow{1}{*}{0.44}&\multirow{1}{*}{0.49}&\multirow{1}{*}{0.21}
\\
\hline
\multirow{1}{*}{10}&\multirow{1}{*}{2,080}&\multirow{1}{*}{312}
&\multirow{1}{*}{0.45}&\multirow{1}{*}{0.45}&\multirow{1}{*}{0.22}
\\
\hline
\multirow{1}{*}{11}&\multirow{1}{*}{1,850}&\multirow{1}{*}{278}
&\multirow{1}{*}{0.46}&\multirow{1}{*}{0.41}&\multirow{1}{*}{0.22}
\\
\hline
\multirow{1}{*}{16}&\multirow{1}{*}{1,250}&\multirow{1}{*}{188}
&\multirow{1}{*}{0.50}&\multirow{1}{*}{0.30}&\multirow{1}{*}{0.50}
\\
\hline
\multirow{1}{*}{17}&\multirow{1}{*}{1,800}&\multirow{1}{*}{270}
&\multirow{1}{*}{0.48}&\multirow{1}{*}{0.28}&\multirow{1}{*}{0.40}
\\
\hline
\multirow{1}{*}{18}&\multirow{1}{*}{2,350}&\multirow{1}{*}{353}
&\multirow{1}{*}{0.52}&\multirow{1}{*}{0.29}&\multirow{1}{*}{0.40}
\\
\hline
\multirow{1}{*}{19}&\multirow{1}{*}{2,250}&\multirow{1}{*}{338}&\multirow{1}{*}{0.43}
&\multirow{1}{*}{0.30}&\multirow{1}{*}{0.50}
\\
\hline
\end{tabular}
\end{table}

The PUE of colocation $\gamma$ is set to 1.6 (typical for colocation), and the default cost for using BES, $\alpha$ is considered 150\$/MWh which we will vary later depending on the BES energy source. 


\begin{table*}[!tp]
\setlength{\abovecaptionskip}{5pt}%
\setlength{\belowcaptionskip}{5pt}%
\caption{The tenant bid data}\label{tab:sm}
\centering
\begin{tabular}{|c|l|l|l|l|l|l|l|l|l|l|l|l|l|l|l|l|l|l|}
\hline
\multirow{2}{*}{\textbf{Hour}}
&\multicolumn{2}{c|}{Tenant\#1}
&\multicolumn{2}{c|}{Tenant\#2}
&\multicolumn{2}{c|}{Tenant\#3}
&\multicolumn{2}{c|}{Tenant\#4}
&\multicolumn{2}{c|}{Tenant\#5}
&\multicolumn{2}{c|}{Tenant\#6}
&\multicolumn{2}{c|}{Tenant\#7}
&\multicolumn{2}{c|}{Tenant\#8}
&\multicolumn{2}{c|}{Tenant\#9}
\\\cline{2-19}
\multirow{1}{*}{}&\multirow{1}{*}{$s_{1}$}&\multirow{1}{*}{$b_{1}$}&\multirow{1}{*}{$s_{2}$}&\multirow{1}{*}{$b_{2}$}&\multirow{1}{*}{$s_{3}$}&\multirow{1}{*}{$b_{3}$}&\multirow{1}{*}{$s_{4}$}&\multirow{1}{*}{$b_{1}$}&\multirow{1}{*}{$s_{5}$}&\multirow{1}{*}{$b_{5}$}&\multirow{1}{*}{$s_{6}$}&\multirow{1}{*}{$b_{6}$}&\multirow{1}{*}{$s_{7}$}&\multirow{1}{*}{$b_{7}$}&\multirow{1}{*}{$s_{8}$}&\multirow{1}{*}{$b_{8}$}&\multirow{1}{*}{$s_{9}$}&\multirow{1}{*}{$b_{9}$}\\\hline
\multirow{1}{*}{5}&\multirow{1}{*}{23}&\multirow{1}{*}{2,737}&\multirow{1}{*}{12}&\multirow{1}{*}{1,284}&\multirow{1}{*}{4}&\multirow{1}{*}{352}&\multirow{1}{*}{67}&\multirow{1}{*}{4,623}&\multirow{1}{*}{26}&\multirow{1}{*}{2,704}&\multirow{1}{*}{5}&\multirow{1}{*}{555}&\multirow{1}{*}{43}&\multirow{1}{*}{3,569}&\multirow{1}{*}{46}&\multirow{1}{*}{5,888}&\multirow{1}{*}{5}&\multirow{1}{*}{570}\\\hline
\multirow{1}{*}{6}&\multirow{1}{*}{33}&\multirow{1}{*}{3,531}&\multirow{1}{*}{18}&\multirow{1}{*}{1,782}&\multirow{1}{*}{6}&\multirow{1}{*}{648}&\multirow{1}{*}{67}&\multirow{1}{*}{8,710}&\multirow{1}{*}{26}&\multirow{1}{*}{2,392}&\multirow{1}{*}{5}&\multirow{1}{*}{470}&\multirow{1}{*}{57}&\multirow{1}{*}{4,617}&\multirow{1}{*}{60}&\multirow{1}{*}{5,100}&\multirow{1}{*}{7}&\multirow{1}{*}{588}\\\hline
\multirow{1}{*}{7}&\multirow{1}{*}{36}&\multirow{1}{*}{3,960}&\multirow{1}{*}{19}&\multirow{1}{*}{1,653}&\multirow{1}{*}{6}&\multirow{1}{*}{462}&\multirow{1}{*}{67}&\multirow{1}{*}{6,700}&\multirow{1}{*}{26}&\multirow{1}{*}{1,976}&\multirow{1}{*}{5}&\multirow{1}{*}{605}&\multirow{1}{*}{30}&\multirow{1}{*}{3,960}&\multirow{1}{*}{32}&\multirow{1}{*}{3,136}&\multirow{1}{*}{4}&\multirow{1}{*}{364}\\\hline
\multirow{1}{*}{8}&\multirow{1}{*}{46}&\multirow{1}{*}{5,336}&\multirow{1}{*}{25}&\multirow{1}{*}{1,950}&\multirow{1}{*}{8}&\multirow{1}{*}{784}&\multirow{1}{*}{66}&\multirow{1}{*}{6,864}&\multirow{1}{*}{25}&\multirow{1}{*}{2,350}&\multirow{1}{*}{5}&\multirow{1}{*}{585}&\multirow{1}{*}{37}&\multirow{1}{*}{4,440}&\multirow{1}{*}{40}&\multirow{1}{*}{4,240}&\multirow{1}{*}{5}&\multirow{1}{*}{565}\\\hline
\multirow{1}{*}{9}&\multirow{1}{*}{48}&\multirow{1}{*}{3,600}&\multirow{1}{*}{26}&\multirow{1}{*}{2,262}&\multirow{1}{*}{8}&\multirow{1}{*}{960}&\multirow{1}{*}{63}&\multirow{1}{*}{8,064}&\multirow{1}{*}{24}&\multirow{1}{*}{1,656}&\multirow{1}{*}{4}&\multirow{1}{*}{476}&\multirow{1}{*}{31}&\multirow{1}{*}{3,162}&\multirow{1}{*}{33}&\multirow{1}{*}{3,564}&\multirow{1}{*}{4}&\multirow{1}{*}{516}\\\hline
\multirow{1}{*}{10}&\multirow{1}{*}{49}&\multirow{1}{*}{4,018}&\multirow{1}{*}{27}&\multirow{1}{*}{3,159}&\multirow{1}{*}{9}&\multirow{1}{*}{792}&\multirow{1}{*}{58}&\multirow{1}{*}{5,510}&\multirow{1}{*}{22}&\multirow{1}{*}{2,112}&\multirow{1}{*}{4}&\multirow{1}{*}{332}&\multirow{1}{*}{33}&\multirow{1}{*}{3,267}&\multirow{1}{*}{35}&\multirow{1}{*}{2,590}&\multirow{1}{*}{4}&\multirow{1}{*}{456}\\\hline
\multirow{1}{*}{11}&\multirow{1}{*}{50}&\multirow{1}{*}{6,000}&\multirow{1}{*}{27}&\multirow{1}{*}{3,078}&\multirow{1}{*}{9}&\multirow{1}{*}{693}&\multirow{1}{*}{53}&\multirow{1}{*}{5,353}&\multirow{1}{*}{20}&\multirow{1}{*}{2,440}&\multirow{1}{*}{4}&\multirow{1}{*}{340}&\multirow{1}{*}{33}&\multirow{1}{*}{3,102}&\multirow{1}{*}{35}&\multirow{1}{*}{3,115}&\multirow{1}{*}{4}&\multirow{1}{*}{384}\\\hline
\multirow{1}{*}{16}&\multirow{1}{*}{55}&\multirow{1}{*}{3,960}&\multirow{1}{*}{30}&\multirow{1}{*}{2,160}&\multirow{1}{*}{10}&\multirow{1}{*}{1,150}&\multirow{1}{*}{39}&\multirow{1}{*}{3,159}&\multirow{1}{*}{15}&\multirow{1}{*}{1,995}&\multirow{1}{*}{3}&\multirow{1}{*}{399}&\multirow{1}{*}{75}&\multirow{1}{*}{7,950}&\multirow{1}{*}{80}&\multirow{1}{*}{8,320}&\multirow{1}{*}{10}&\multirow{1}{*}{1,290}\\\hline
\multirow{1}{*}{17}&\multirow{1}{*}{52}&\multirow{1}{*}{4,420}&\multirow{1}{*}{28}&\multirow{1}{*}{2,016}&\multirow{1}{*}{9}&\multirow{1}{*}{1,170}&\multirow{1}{*}{36}&\multirow{1}{*}{4,068}&\multirow{1}{*}{14}&\multirow{1}{*}{1,484}&\multirow{1}{*}{2}&\multirow{1}{*}{178}&\multirow{1}{*}{60}&\multirow{1}{*}{7,620}&\multirow{1}{*}{64}&\multirow{1}{*}{4,352}&\multirow{1}{*}{8}&\multirow{1}{*}{952}\\\hline
\multirow{1}{*}{18}&\multirow{1}{*}{57}&\multirow{1}{*}{7,581}&\multirow{1}{*}{31}&\multirow{1}{*}{3,100}&\multirow{1}{*}{10}&\multirow{1}{*}{1,030}&\multirow{1}{*}{37}&\multirow{1}{*}{4,144}&\multirow{1}{*}{14}&\multirow{1}{*}{1,358}&\multirow{1}{*}{2}&\multirow{1}{*}{212}&\multirow{1}{*}{60}&\multirow{1}{*}{6,540}&\multirow{1}{*}{64}&\multirow{1}{*}{7,360}&\multirow{1}{*}{8}&\multirow{1}{*}{712}\\\hline
\multirow{1}{*}{19}&\multirow{1}{*}{47}&\multirow{1}{*}{4,136}&\multirow{1}{*}{25}&\multirow{1}{*}{3,275}&\multirow{1}{*}{8}&\multirow{1}{*}{744}&\multirow{1}{*}{39}&\multirow{1}{*}{5,031}&\multirow{1}{*}{15}&\multirow{1}{*}{1,380}&\multirow{1}{*}{3}&\multirow{1}{*}{270}&\multirow{1}{*}{75}&\multirow{1}{*}{7,725}&\multirow{1}{*}{80}&\multirow{1}{*}{5,760}&\multirow{1}{*}{10}&\multirow{1}{*}{930}\\\hline
\end{tabular}
\end{table*}

\subsection{Experiment and Result Analysis}

To evaluate the performance of the MEDR method, we consider the approximation ratio of the FPTAS mechanism, tenants' utility and social cost reductions.
For simplification, we prepare all the simulation data and related parameter data of algorithms together into a file which is used as input for the bidding decision program. During the test running process, our bidding decision program will read the data in a batch, process all the algorithms execution, and then output all results to a file for analysis.   

\subsubsection{Approximation Ratios}
Our FPTAS mechanism can achieve $1+\epsilon$ approximation ratio. Though FPTAS ensures the running time is polynomial in  $n$, and  $1/\epsilon$. 
To evaluate the efficiency of the FPTAS mechanism in practice, we use approximation ratio metric which is concluded by the ratio between the FPTAS-approximation algorithm and the DOPT optimal, i.e. $AR=y(FPTAS)\//y(DOPT)$. 
In the simulation, we vary the parameters $\alpha,\gamma,$ and $\epsilon$ to run the FPTAS and DOPT algorithm to obtain the results respectively. 
In each test, we fix two parameters and change one parameter regularly. 
Obviously we have three cases of ratio testing, which are named as $\alpha$-ratio test, $\gamma$-ratio test and $\epsilon$-ratio test, respectively. In $\alpha$-ratio test, we set $\gamma$ to $1.6$ and set $\epsilon$ to $0.5$ and vary the parameter $\alpha$ with a range from 140 to 320 with an increasing step 20. In $\gamma$-ratio test, $\epsilon$ is set to $0.5$, $\alpha$ is set to 180\$, and $\gamma$ is changed from 1.1 to 2 with an increasing step $0.1$. In $\epsilon$-ratio test, $\alpha$ is set to $180$\$, $\gamma$ is set to $1.6$, and $\gamma$ is changed from 1.1 to 2 with an increasing step $0.1$. 

We run the tests for all the $11$ EDR reduction instances. The results are shown in Fig.~\ref{fig:ratio}. We observe that all ratios are between $1$ and $1+\epsilon$. Moreover, the ratios are closer to the line with ratio $1$. It means that the FPTAS solution is very close to the optimal solution. 
Most of the line with ratios of varying $\alpha$ seems to be more flat which means the $\alpha$ parameter has little impact on the ratios. There are several broken lines in Fig.~\ref{fig:gama} and Fig.~\ref{fig:eps}. It demonstrates that the ratio exists a special sensitivity to parameters $\gamma$ and $\epsilon$. 
The largest impact factor is the parameter $\epsilon$. With the increasing of $\epsilon$, the ratio increases as well as in theory.

\begin{figure*}[t]
  \centering
  \mbox{
     \subfigure[$\alpha$-ratio test: the ratio between FPTAS and DOPT for eleven EDR instances on varying $\alpha$, where $\gamma=1.6,\epsilon=0.5$]{
        \includegraphics[width=0.33\textwidth,height=0.25\textwidth]{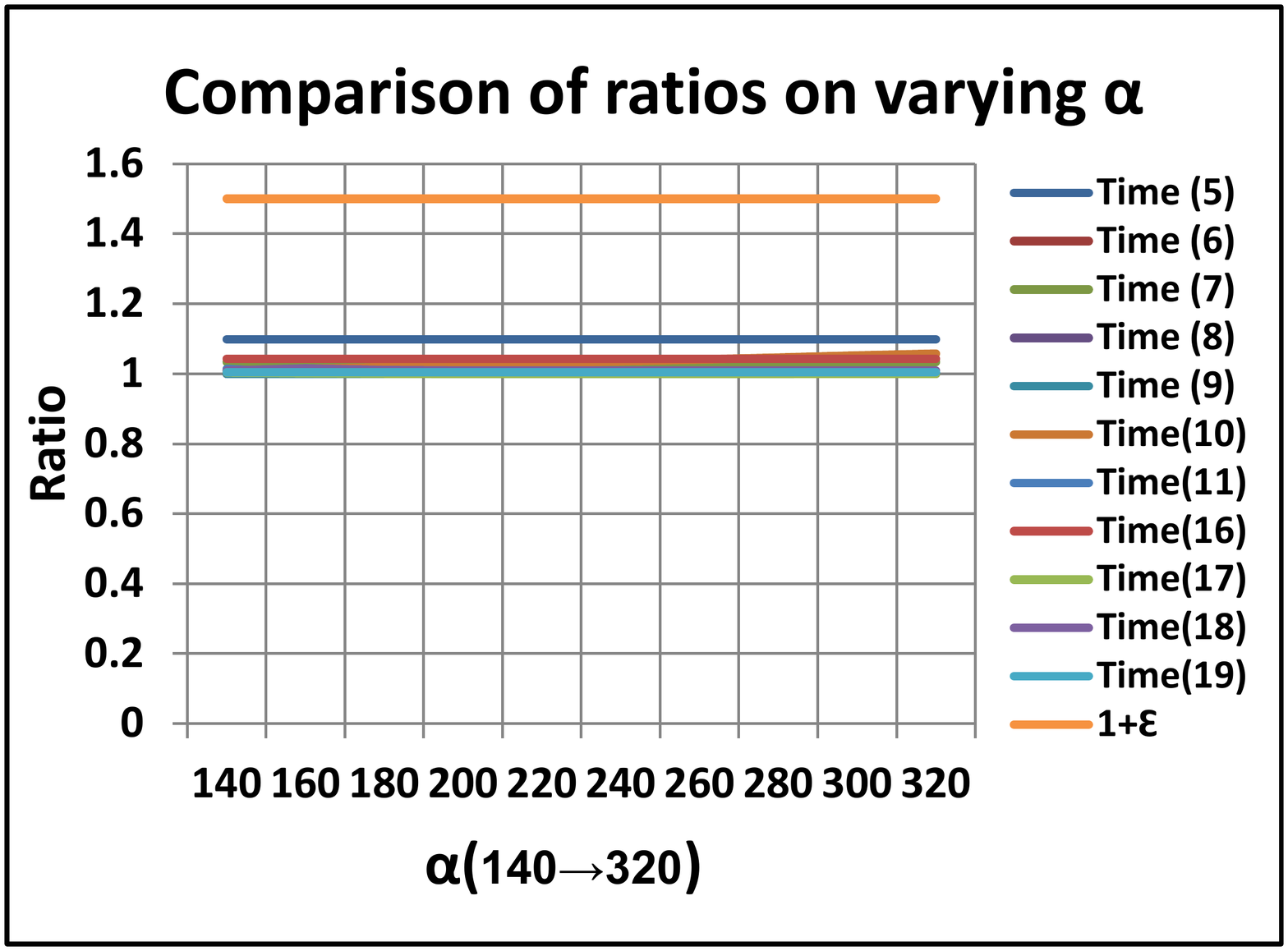}\label{fig:alpha}}\hspace{.02in}

      \subfigure[$\gamma$-ratio test: the ratio between FPTAS and DOPT for eleven EDR instances on varying $\gamma$, where $\alpha=180\$,\epsilon=0.5$]{
    \includegraphics[width=0.33\textwidth,height=0.25\textwidth]{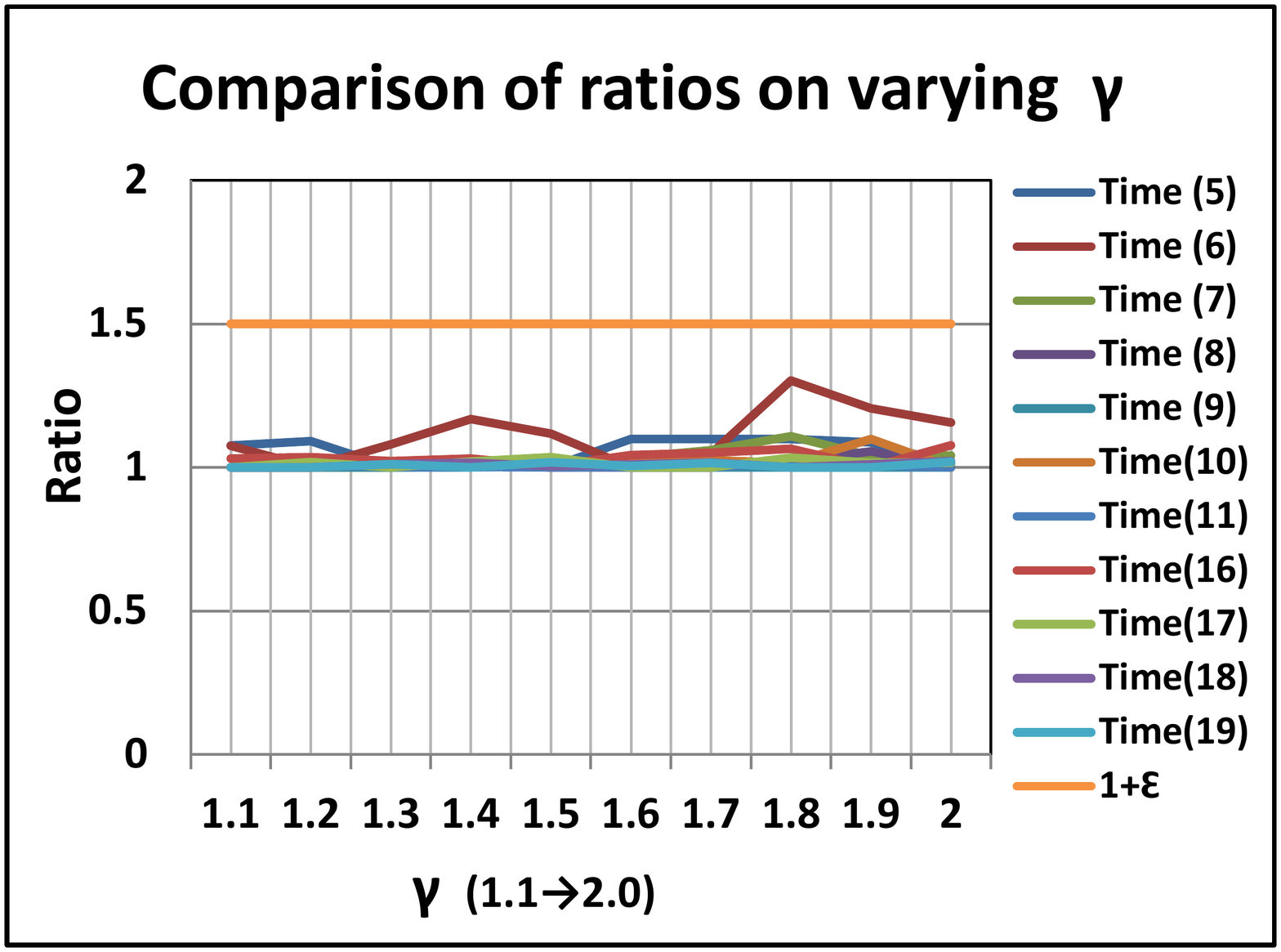}\label{fig:gama}}\hspace{.02in}
       
       \subfigure[$\epsilon$-ratio test: the ratio between FPTAS and DOPT for eleven EDR instances on varying $\epsilon$, where $\alpha=180\$,\gamma=1.6$]{\includegraphics[width=0.33\textwidth,height=0.25\textwidth]{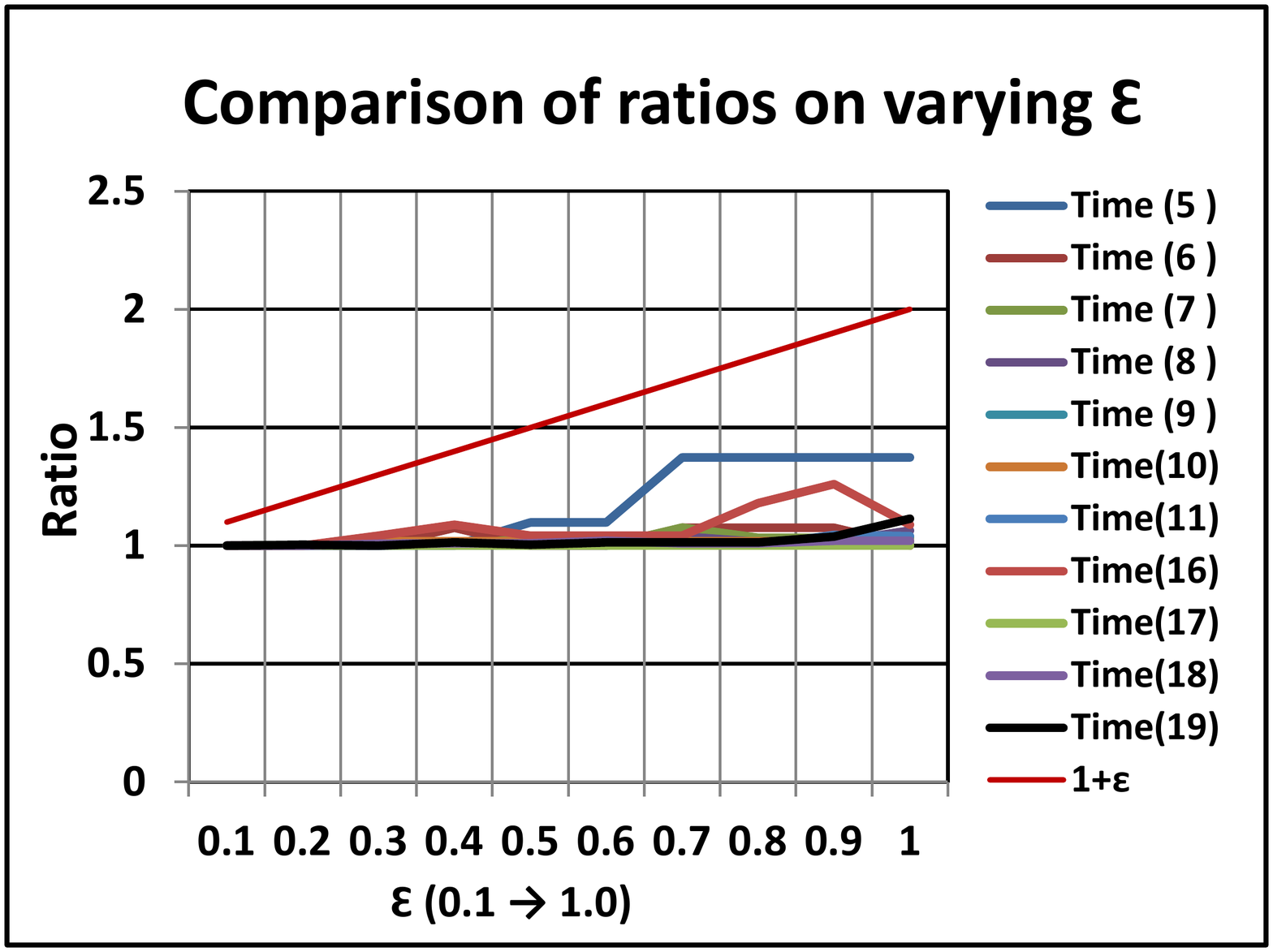}\label{fig:eps}}
    }
\caption{The performance evaluation on approximation ratios in distinct time of EDR instances.}
  \label{fig:ratio}
\end{figure*}

\begin{figure*}[t]
  \centering

\mbox{

      \subfigure[Comparison of tenants' non-negative utilities when $\alpha=180,\gamma=1.6,\epsilon=0.5$.]
      {\includegraphics[width=0.33\textwidth,height=0.25\textwidth]{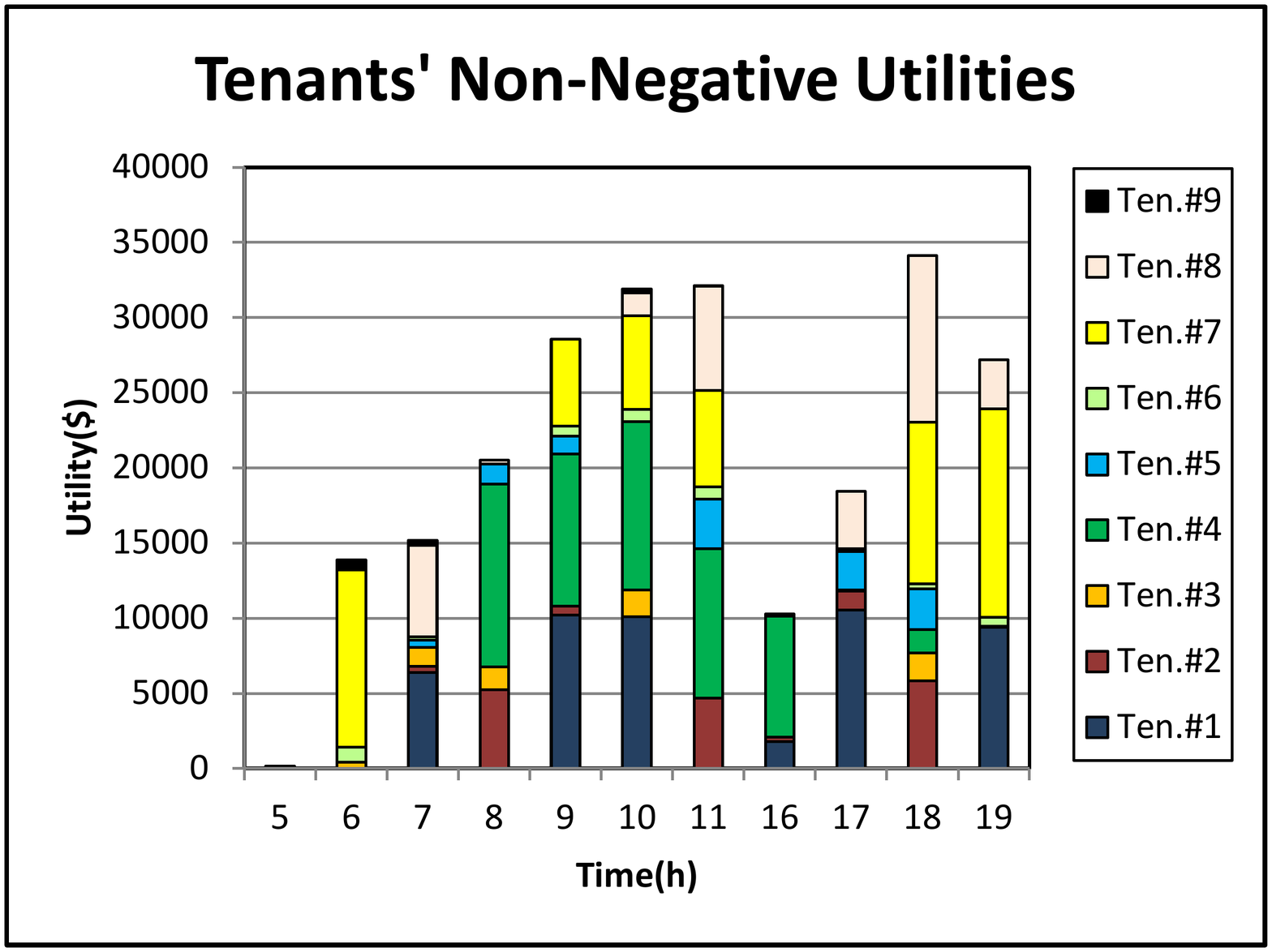}\label{fig:util}}\hspace{.02in}
       \subfigure[Comparison of social cost ratio between FPTAS and BES, where $\gamma=1.6$, $\epsilon=0.5$ and $\alpha$ varying from 140 to 320 with a growing step 20.]{\includegraphics[width=0.33\textwidth,height=0.25\textwidth]{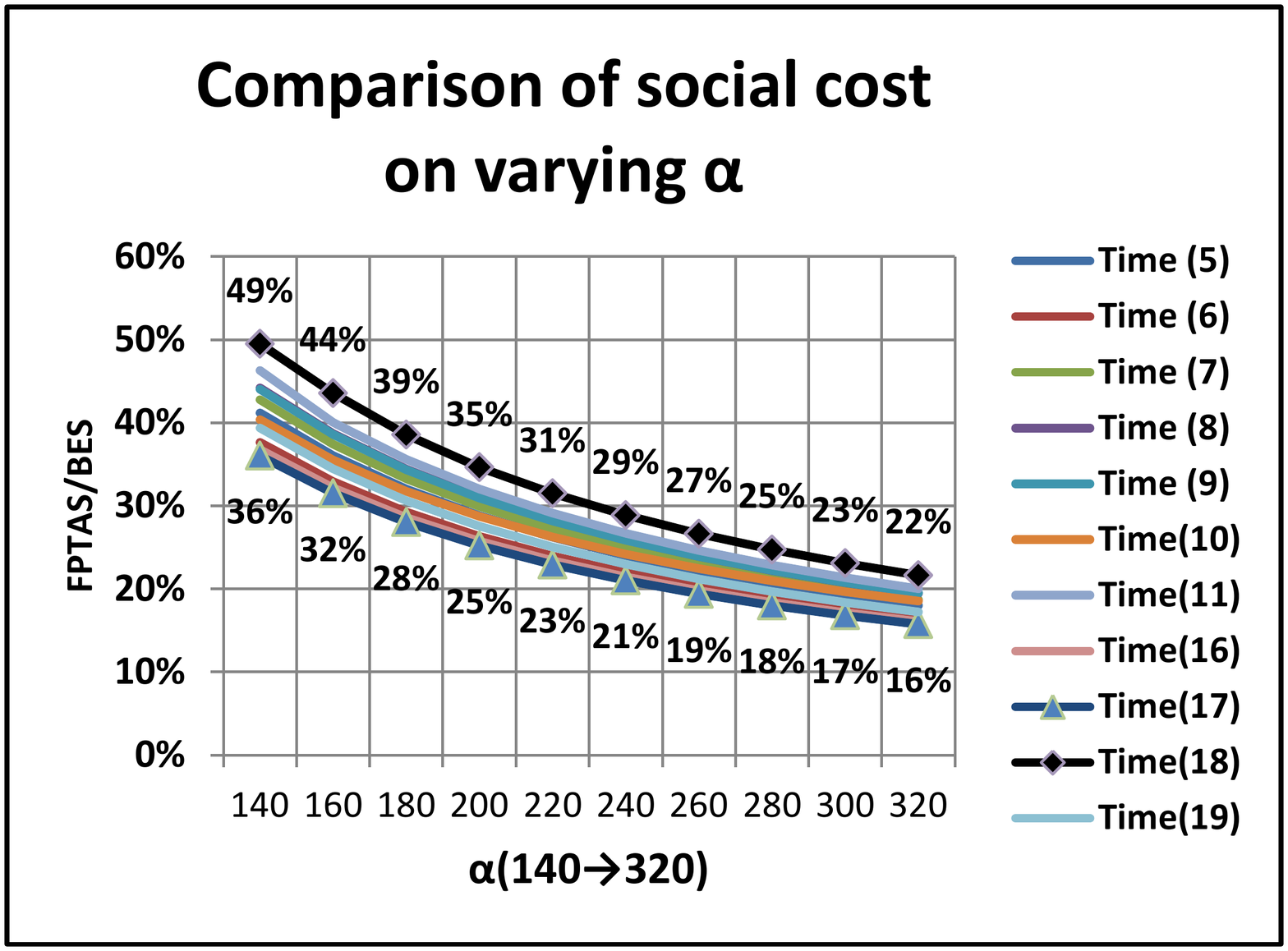}\label{fig:scralpha}}\hspace{.02in}
       \subfigure[Comparison of social cost ratio between FPTAS and BES, where $\alpha=180$, $\epsilon=0.5$ and $\gamma$ varying from 1.1 to 2.0 with a growing step 0.1.]{\includegraphics[width=0.33\textwidth,height=0.25\textwidth]{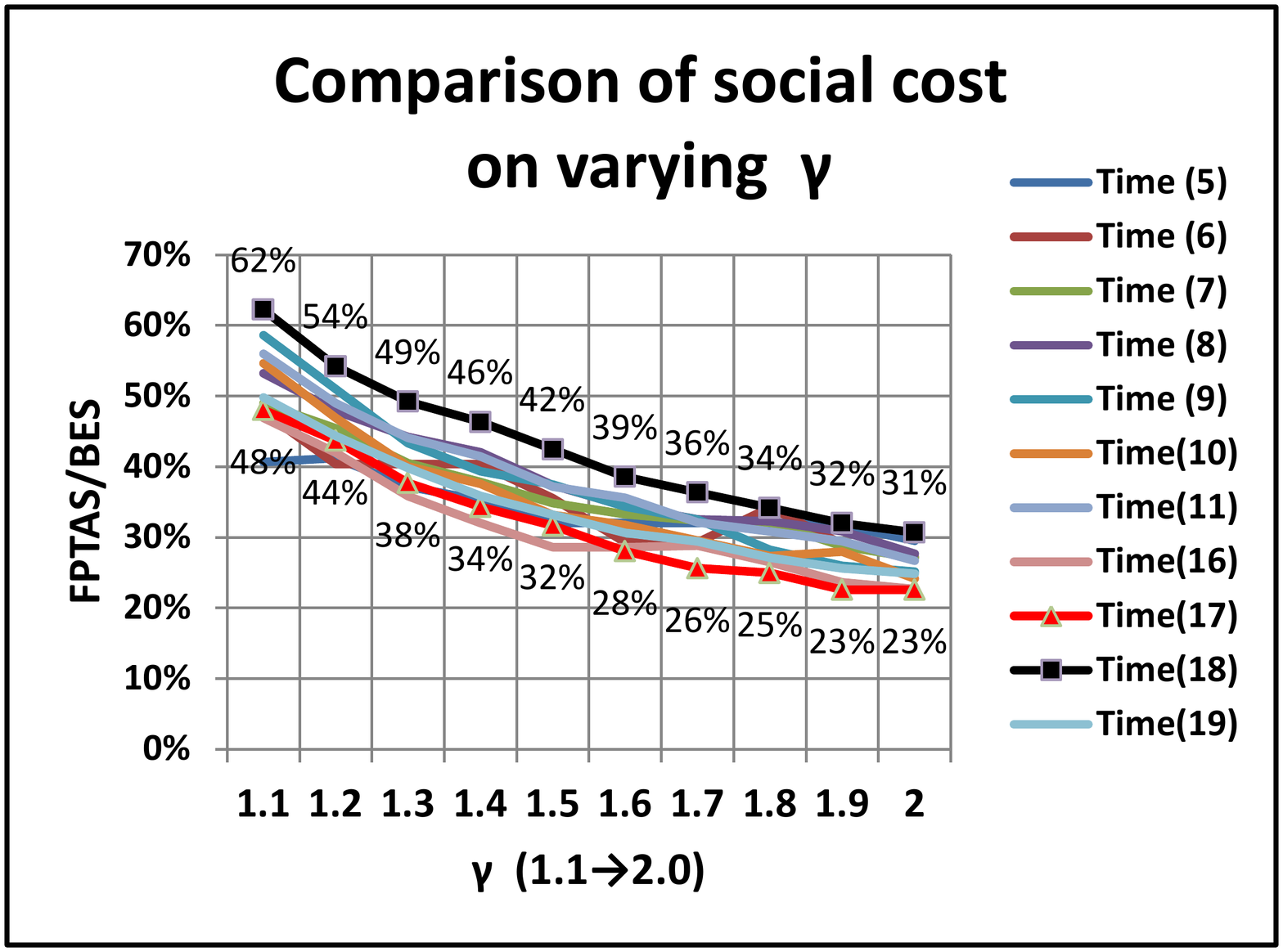}\label{fig:scrgama}}
    }
  
  \caption{The performance evaluation on tenants' utility and social costs.}
  \label{fig:pefeval}
\end{figure*}


\subsubsection{Agents' Utilities}
We study each agent's utility in all the experiments. 
Note that agents' utilities are concluded from the payment for each tenant subtracts his/her actual cost. 
Actually, we got that each tenant obtains a non-negative utility in our experiment. 
We only show the result in Fig.~\ref{fig:util} by letting $\alpha=180\$,\gamma=1.6, \epsilon=0.5$. 
We illustrate each winner tenant's utility in each hour.  
Fig.~\ref{fig:util} shows that the case of the 18th hour time has the largest utility and the 5th hour time has the minimum utility.

\subsubsection{Social Cost Reduction Compared to BES only}
Each winner tenant has obtained a non-negative utility, which implies that the colocation operator will pay a lot of money to these tenants. To study whether this payment is too much, we investigate the social costs when we only use BES. The results are illustrated in Fig.~\ref{fig:scralpha} and Fig.~\ref{fig:scrgama} when the parameters $\alpha$ and $\gamma$ vary. The results reveal that the social costs given by the FPTAS mechanism is much smaller than that one of BES only. The percentages are declined when $\alpha$ or $\gamma$ increases.
So it's very important for the colocation data center to enable more tenants to attend the energy demand response activities.
 
\section{Related works}~\label{sec:rw}
Besides the work in~\cite{zhang2015infocom}, there are several existing research on mechanism design on DR. 
Ren and Islam~\cite{ren14colocation} studied the mechanism design for colocation demand response, but their mechanism is not truthful and may not meet the target of EDR. 
Chen et al.~\cite{chen2015greening} studied the green colocation data centers by designing a pricing mechanism to fulfill energy reduction requirement for EDR. 
The energy reduction from tenants is calculated by the price-taking and price-anticipating equilibrium. 
Zhou et al.~\cite{zhou2015smart} studied demand resource on geo-distributed cloud through VCG-based mechanism, in which the utility of each agent depends highly on its interactive workload. Sun et al.~\cite{sunfair2015} considered fairness among the mechanism design and provided online mechanism with competitive ratio of 3.2 in expectation. 
Ahmed et al.~\cite{ahmed2015contract} proposed a contract-based mechanism, in which the colocation operator offers a set of contracts (i.e., a pair of energy reduction and rewards) to tenants and tenants can voluntarily select none or one of the contracts to accept, while the objective is to minimize the operator's cost, the sum of rewards plus the cost of BES. Islam et al.~\cite{islam2015paying} reduced the operator's cost by learning tenants' response to reward.
Tran et al.~\cite{tran2015incentive} used  a two-stage Stackelberg game to  model the economic demand response where the operator can adjust an elastic energy reduction target.

A closely related work is DR in smart grids. Zhou et al.~\cite{zhou2015demand} studied the  mechanism design on DR in smart grids. Let $\alpha = 1$, and there is an upper bound on the BES, i.e. $ y \le z_{\max}$. 
A randomized FPTAS mechanism was given in~\cite{zhou2015demand}. Their idea is to combine with smooth analysis and randomize auction. Actually, Dough and Roughgarden~\cite{dughmi2014black} showed that if there exists an FPTAS approximation, then this algorithm can be transformed into a truthful in expectation mechanism that retains the FPTAS property. The work in~\cite{dughmi2014black} does not require the existence of FPTAS. However, it still remains open whether there exists a deterministic FPTAS. 


Briest et al.~\cite{briest2011approximation} presented a truthful FPTAS for the max-knapsack problem. 
Our problem differs from the min-knapsack in which we have BES such that the capacity of the knapsack we need to cover is soft. 

\section{Conclusions}\label{sec:conc}

In this paper, we have proposed a deterministic truthful FPTAS mechanism with 1+$\epsilon$ approximation ratio for a reverse auction of EDR in colocation data centers. We have developed an auction system and implemented a bidding decision tool for simulation experiments. The experimental results demonstrate the effectiveness of our methods. 
In future work, we study the more practical utility of our mechanism and the performance optimization with parallelization techniques for the dynamical programming. Besides, the provided technique allows us to deal with single minded agents, and both the size of energy and the cost are private information. Many open problems arise in the area of demand response. For example, agents are multi-minded. 
\section*{Acknowledgment}
This research was in part supported by NSFC (11671355, 61772466, U1836202, 60970125, 61272303), the “National Key R\&D Program of China”(2017YFB1401300, 2017YFB1401304), the Zhejiang Provincial Natural Science Foundation for Distinguished Young Scholars under No. LR19F020003, the Provincial Key Research and Development Program of Zhejiang, China under No. 2017C01055, and the Alibaba-ZJU Joint Research Institute of Frontier Technologies.

\bibliographystyle{IEEEtran}
\bibliography{md}
\end{document}